\documentclass[12pt]{article}
\usepackage{graphicx}
\usepackage{amssymb}
\usepackage{amsmath}
\usepackage{amsfonts}
\usepackage{amsthm}
\usepackage[english]{babel}
\usepackage[latin1,ansinew]{inputenc}
\usepackage{a4wide}
\usepackage{color}
\newcommand{\be}{\begin{eqnarray}}
\newcommand{\ee}{\end{eqnarray}}
\newtheorem{theo}{Theorem}
\newtheorem{prop}{Proposition}

\newcommand{\bu}{{\bf u}}

\begin{document}

\title{A Causal Formulation of\\
Dissipative Relativistic Fluid Dynamics\\ with or without Diffusion}

%


\author{\it Heinrich Freist\"uhler\thanks{Department of Mathematics, University of Konstanz, 
Konstanz, Germany}}

\date{November 1, 2022}

\maketitle

\begin{abstract}
The article proposes a causal five-field formulation of dissipative relativistic fluid dynamics
as a quasilinear symmetric hyperbolic system of second order.  The system is determined by four
dissipation coefficients $\eta,\zeta,\kappa,\mu$, free functions of the fields, which quantify
shear viscosity, bulk viscosity, heat conductivity, and diffusion. 
\end{abstract}

\newpage
\section{Introduction}
The present note is related to the question of whether dissipative relativistic 
fluid dynamics can be properly modeled by a \textit{hyperbolic, causal five-field theory} 
\be \label{nsf}
\begin{aligned}
                 \frac{\partial}{\partial x^{\beta}}\,(T^{\alpha \beta}+\Delta T^{\alpha \beta})=0,
       \quad     \quad     \frac{\partial}{\partial x^{\beta}}(N^{\beta}+\Delta N^\beta)=0,
                \end{aligned}
\ee
where the dissipation-free parts of the energy-momentum tensor\footnote{We work with 
the Minkowski metric $g^{\alpha\beta}$ of signature $(-,+,+,+)$.} and 
the particle  number density current,
\be
T^{\alpha\beta}=(\rho+p)U^{\alpha}U^{\beta}+pg^{\alpha\beta},\quad
N^\beta=nU^\beta,
\label{TN}
\ee
are given in terms of the fluid's velocity $U^\alpha$ and its 
energy density $\rho$, pressure $p$, and particle number density $n$ ---
one speaks of five fields since $U_\alpha$ is constrained by $U^\alpha U_\alpha=-1$ and
$\rho,p,n$ are related to each other by an equation of state ---, and 
the {\it dissipation tensors}\ $\Delta T^{\alpha\beta}, \Delta N^\beta$ are linear 
in the space-time gradients of these fields (``relativistic Navier-Stokes'').
That question has first\footnote{An interesting different line of thinking was later 
started in \cite{BDN}.}
been answered in the affirmative in the articles \cite{FTPR,FTIG}, to which this note 
is a supplement. In the introduction to his ``Hyperbolic Conservation Laws in Continuum Physics'' 
\cite{Da}, Dafermos suggests that 
``the umbilical cord joining continuum physics with the theory of partial differential equations 
should not be severed, as it is still carrying nourishment in both directions''. 
This paper is part of an attempt to contribute in that spirit.

Its primary purpose is to include diffusion in addition to viscosity
and heat conduction. As in \cite{FTIG,FTBF}, we find corresponding equations from the three guiding principles that
(a) when written in the natural Godunov variables that make the dissipation-free part
symmetric hyperbolic in the first-order sense, they should be symmetric hyperbolic in the sense
of (a covariant version of) the Hughes-Kato-Marsden (HKM) conditions, (b) they should be first-order
equivalent to the classical theories by Eckart and Landau, and (c) all signal speeds should be
bounded by the speed of light. A side-motivation is related to the different treatment of heat 
conduction in \cite{FTIG,FTBF}. While the paper \cite{FTIG} on general non-barotropic fluids
interprets, like Landau, heat conduction as an effect within
matter conservation \eqref{nsf}$_2$, the article \cite{FTBF} on barotropic fluids
represents it, like Eckart, within the conservation laws of energy and momentum \eqref{nsf}$_1$. 
This latter option is however also available for non-barotropic fluids, and it is the one we adopt 
in this paper. 

In Section 2, we state the dissipation tensors $\Delta T^{\alpha\beta}, \Delta N^\beta$ 
that are determined by the coefficients $\eta, \zeta,\chi,\mu$ of shear viscosity, bulk viscosity, 
heat conductivity, and diffusion. Section 3
shows that the resulting five-field theory is first-order equivalent with the corresponding 
Eckart and Landau formulations. In Section 4 we demonstrate that it belongs to the 
HKM class and is causal. 
Sec.\ 5 establishes the validity of the second law of thermodynamics to leading order.
We assume at most places (though not in Sec.\ 5) that the fluid is polytropic,
\be
\rho=mn+\frac p{\gamma-1},\quad \text{with }m>0\text{ and }1<\gamma<2. 
\label{poly}
\ee
As in \cite{FTIG}, this assumption is made for concreteness only; the essence of the argumentation does not 
depend on it, and corresponding results for arbitrary other massive non-barotropic fluids 
follow through obvious small adaptations. 

For the case that all dissipative mechanisms are active, i.\ e., \eqref{etachipos}, \eqref{mupos},  the model 
proposed and discussed in this paper has been studied by Sroczinski who showed in \cite{S20} that for initial data
that are sufficiently small perturbations of a spatially homogeneous state, the Cauchy problem has a unique solution
for all times $t>0$ that decays for $t\to\infty$ to the reference state. The same is shown in \cite{FS} for the case 
with \eqref{etachipos} and $\mu=0$, i.\ e., in the absence of diffusion. 

While the approach pursued here seems mathematically coherent and plausible, its physical validity is to be 
further examined, notably in view of the fact that the second law of thermodynamics is so far established only 
to leading order.
A different, elaborate formulation of dissipative fluid dynamics in terms of hyperbolic balance laws
is Rational Extended Thermodynamics,
which does not have the latter problem; cf.\ the recent monograph \cite{RuSu} and references therein.

\section{Proposed new formulation and main results}
\setcounter{equation}0
Besides fixed choices of coefficients of shear viscosity, bulk viscosity,  
and heat conduction,
\be\label{etachipos}
\eta>0,\quad \zeta\ge 0,\quad\text{and}\quad\chi>0,
\ee
we now consider also a diffusion coefficient 
\be\label{mupos}
\mu>0;
\ee
all four of these coefficients can be taken to be widely arbitrary functions of 
the thermodynamic variables. 
Our proposal is to use\footnote{We utilize
$\Pi^{\alpha\beta}=g^{\alpha\beta}+U^\alpha U^\beta$.} 
\begin{eqnarray}
-\Delta T^{\alpha\beta}
&=&
\eta\Pi^{\alpha\gamma}\Pi^{\beta\delta}
\left({\partial U_\gamma\over \partial x^\delta}
+
{\partial U_\delta\over \partial x^\gamma}
-{2\over 3}g_{\gamma\delta}
{\partial U^\epsilon\over \partial x^\epsilon}\right)
+\tilde\zeta\Pi^{\alpha\beta}
{\partial U^\gamma\over \partial x^\gamma}\notag
\\
\label{DeltaT}
 &&
+\sigma\left( U^\alpha U^\beta 
{\partial U^\gamma\over \partial x^\gamma}
-\left(g^{\alpha\gamma} U^\beta + g^{\beta\gamma}U^\alpha\right) 
U^\delta
{\partial U_\gamma\over \partial x^\delta}\right)
\\
&&
+\chi\left(U^\alpha\frac{\partial\theta}{\partial x^\beta} 
+
U^\beta\frac{\partial\theta}{\partial x^\alpha}
-
g^{\alpha\beta} 
U^\gamma\frac{\partial\theta}{\partial x^\gamma}
\right)\notag
\end{eqnarray}
and
\begin{eqnarray}\label{DeltaN}
 -\Delta N^\beta
&\equiv&
\mu
g^{\beta\delta}
{\partial \psi\over \partial x^\delta}
+
\tilde\sigma\left(
U^\beta 
{\partial U^\gamma\over \partial x^\gamma}
-U^\gamma
{\partial U^\beta\over \partial x^\gamma}
\right)
\end{eqnarray}
or, written as matrices with respect to the fluid's rest frame,\footnote{We
write $\bu$ for the 3-velocity with respect to the fluid's rest frame at a given point. 
(While $\bu=0$ at that point, its gradient is free.) $\dot{\phantom{x}}$ means derivative
with respect to $x^0$, $\nabla$ derivatives with respect to $(x^1,x^2,x^3)$, all in the rest frame.
}
\be\label{matrixT}
-\Delta T|_0=
\begin{pmatrix}
-\chi\dot\theta+\sigma \nabla\cdot \bu &\chi\nabla\theta-\sigma \dot\bu^\top\\
\chi\nabla^\top\theta-\sigma \dot\bu&\eta{\bf S}\bu+(
\tilde\zeta\nabla\cdot \bu -\chi\dot\theta)\bf I
\end{pmatrix}
\ee
and
\be\label{matrixN}
-\Delta N|_0=
\begin{pmatrix}
-\mu\dot\psi+
\tilde\sigma
\nabla\cdot \bu\
&\ \mu\nabla\psi-
\tilde\sigma
\dot\bu^\top
\end{pmatrix},
\ee
where 
$$
\theta,\quad \psi= h/\theta-s  
$$
denote temperature and Israel's thermal potential, with $s$ and $h=(\rho+p)/n=$   
the specific entropy and specific enthalpy ($\psi=g/\theta$ with $g=h-\theta s$ the chemical potential)..
The coefficients $\sigma,\tilde\zeta,$ and $\tilde \sigma$ are given by
\be\label{choicetildezeta}
\sigma=\frac43\eta+\tilde\zeta,
\quad
\tilde\zeta=\zeta+\tilde\zeta_1+\tilde\zeta_2+\tilde\zeta_3,
\quad
\tilde\sigma=(\sigma+\chi\theta)/h
\ee
with\footnote{This implies that 
$\sigma=((4/3)\eta+\zeta+\tilde\zeta_1+\tilde\zeta_3)/(1-(\gamma-1)(1-m/h))$.}
\be\label{choicezeta1zeta2zeta3}
\tilde\zeta_1=-(\gamma-1)\left(2-\gamma+m/h\right)\chi\theta   ,
\quad
\tilde\zeta_2=(\gamma-1)(1-m/h)\sigma, 
\quad 
\tilde\zeta_3=(\gamma-1)^2(m^2/\theta)\mu.
\ee
The following will be established in Secs.\ 3, 4, 5.
\begin{theo}
The model \eqref{nsf}--\eqref{poly}, \eqref{DeltaT}--\eqref{choicezeta1zeta2zeta3} is 
first-order equivalent with the Eckart and Landau models with 
diffusion. It is symmetric hyperbolic 
when written in Godunov-Boillat variables. It is causal when $\tilde\zeta\ge-\frac13\eta$, and sharply causal 
when $\tilde\zeta=-\frac13\eta$.\footnote{This can be expressed as a (natural) restriction on $\chi$ in dependence
on $\eta,\zeta$ and $\mu$: $\chi\le\chi_*$ or $\chi=\chi_*$, respectively,
with a function $\chi_*=\chi_*(\eta,\zeta,\mu)>0$.}   
The entropy production in this model is non-negative to leading (first) order in the dissipation coefficients $\eta,\zeta,\chi,\mu$. 
\end{theo}

\section{First-order equivalence with Eckart and Landau}
\setcounter{equation}0
Paper \cite{FTIG} contains a formal description of mappings that transform between different versions of the equations 
of motion which are equivalent up to errors that are quadratic in the dissipation coefficients; notably, the 
descriptions of Eckart \cite{eck} and Landau \cite{LL} are first-order equivalent in this sense.
Using that formalism, first-order equivalences are composed from modifications called velocity shifts, thermodynamic 
shifts, and gradient reexpressions. Suppressing diffusion for a moment, i.e., briefly replacing \eqref{mupos} by 
$\mu=0$, we start from the usual Eckart tensor, in rest frame notation       
\be\label{matrixEckart}
\begin{pmatrix}
0&\chi\left(\nabla\theta+\theta\dot\bu^\top\right)\\
\chi\left(\nabla^\top\theta+\theta\dot\bu\right)&\eta \bf S\bu+\zeta \nabla\cdot \bu\,\bf I\\
0&0 
\end{pmatrix}.
\ee
We combine first a velocity shift (cf.\ \cite{FTIG}, p.\ 11,  first ``assignment'')
\be\notag
\begin{pmatrix}
0&\chi\left(\nabla\theta+\theta\dot\bu^\top\right)\\
\chi\left(\nabla^\top\theta+\theta\dot\bu\right)&\eta \bf S\bu+\zeta \nabla\cdot \bu\,\bf I\\
0&0 
\end{pmatrix}
\to
\begin{pmatrix}
0&\chi\nabla\theta\\
\chi\nabla^\top\theta&\eta {\bf S}\bu+\zeta \nabla\cdot \bu\,\bf I\\
0&-\left(\chi\theta/h\right) \dot\bu^\top
\end{pmatrix}.
\ee
and a thermodynamic shift (cf.\ \cite{FTIG}, p.\ 11, second ``assignment'' with (4.9) and (6.3))
\be\notag
\begin{pmatrix}
0&\chi\nabla\theta\\
\chi\nabla^\top\theta&\eta{\bf S}\bu+\zeta \nabla\cdot \bu\,\bf I\\
0&-\left(\chi\theta/h\right) \dot\bu^\top
\end{pmatrix}
\to
\begin{pmatrix}
-\chi\dot\theta&\chi\nabla\theta\\
\chi\nabla^\top\theta&\eta{\bf S}\bu+((\zeta+\tilde\zeta_1)\nabla\cdot \bu -\chi\dot\theta)\bf I\\
\left(\chi\theta/h\right)\nabla\cdot \bu&-\left(\chi\theta/h\right) \dot\bu^\top
\end{pmatrix}
\ee
with 
$$
\tilde\zeta_1=-(\gamma-1)\left(2-\gamma+\frac mh\right)\chi\theta,
$$
and then follow a second velocity shift
\be\notag
\begin{pmatrix}
-\chi\dot\theta&\chi\nabla\theta\\
\chi\nabla^\top\theta&\eta{\bf S}\bu+((\zeta+\tilde\zeta_1)\nabla\cdot \bu -\chi\dot\theta)\bf I\\
\left(\chi\theta/h\right)\nabla\cdot \bu&-\left(\chi\theta/h\right) \dot\bu^\top
\end{pmatrix}
\to
\begin{pmatrix}
-\chi\dot\theta&\chi\nabla\theta-\sigma \dot\bu^\top\\
\chi\nabla^\top\theta-\sigma \dot\bu&\eta{\bf S}\bu+((\zeta+\tilde\zeta_1)\nabla\cdot \bu -\chi\dot\theta)\bf I\\
\left(\chi\theta/h\right)\nabla\cdot \bu&-\left((\chi\theta+\sigma)/h\right) \dot\bu^\top
\end{pmatrix}
\ee
by another thermodynamic shift, which leads to  
\be\label{tensorwithoutdiffusion}
\begin{pmatrix}
-\chi\dot\theta+\sigma \nabla\cdot \bu &\chi\nabla\theta-\sigma \dot\bu^\top\\
\chi\nabla^\top\theta-\sigma \dot\bu&\eta{\bf S}\bu+((\zeta+\tilde\zeta_1+\tilde\zeta_2)\nabla\cdot \bu -\chi\dot\theta)\bf I\\
\left((\chi\theta+\sigma)/h\right)\nabla\cdot \bu&-\left((\chi\theta+\sigma)/h\right) \dot\bu^\top
\end{pmatrix}
\ee
with 
$$
\tilde\zeta_2=(\gamma-1)\left(1-\frac mh\right)\sigma.
$$
The matrix in \eqref{tensorwithoutdiffusion} is already the rest-frame form in \eqref{matrixT}, \eqref{matrixN} 
in the case $\mu=0$.

To now include also diffusion, recall that according to Kluitenberg, de Groot, Mazur \cite{KGM} it is represented in
the Eckart frame by 
$$
-\Delta N^\beta=\mu\Pi^{\beta\gamma}\frac{\partial \psi}{\partial x^\gamma};
$$
i.\ e., we have to superimpose on \eqref{matrixEckart} the matrix 
$$
\begin{pmatrix}
0&0\\
0&0\\
0&\mu\nabla\psi 
\end{pmatrix}.
$$
From this starting point, a thermodynamic shift and a gradient reexpression 
lead to 
\be\notag
\to
\begin{pmatrix}
0&0\\
0&-m(\gamma-1)\dot\psi\\
-\mu\dot\psi&\mu\nabla\psi 
\end{pmatrix}
\to
\begin{pmatrix}
0&0\\
0&\tilde\zeta_3 \nabla\cdot \bu \\
-\mu\dot\psi&\mu\nabla\psi 
\end{pmatrix}
\ee
with
$$
\tilde\zeta_3=(\gamma-1)^2\frac{m^2}\theta\mu 
$$
Adding the last version to \eqref{tensorwithoutdiffusion}, we reach   
yields 
\eqref{matrixT} - \eqref{choicezeta1zeta2zeta3}.

\section{Symmetric hyperbolicity and causality}
\setcounter{equation} 0
We start from the general equivariant forms of tensors 
$-\Delta T^{\alpha\beta}$ and $-\Delta N^\beta$ that are linear in the gradients of 
the state variables. These forms are  (cf.\ \cite{FTIG})
$$ 
-\Delta T^{\alpha\beta}
\equiv
U^\alpha U^\beta P
+(\Pi^{\alpha\gamma} U^\beta + \Pi^{\beta\gamma}U^\alpha) Q_\gamma
+\Pi^{\alpha\beta}R 
+\Pi^{\alpha\gamma}\Pi^{\beta\delta}S_{\gamma\delta}
\label{ournewTansatzg}
$$ 
with
$$
P
=
\tau U^\gamma{\partial \theta\over \partial x^\gamma}
+
\sigma{\partial U^\gamma\over \partial x^\gamma}
+\check\iota U^\gamma{\partial \psi\over \partial x^\gamma},
\quad
Q_\gamma
\equiv
\nu{\partial \theta\over \partial x^\gamma}
+\check\varsigma U^\delta
{\partial U_\gamma\over \partial x^\delta}
+\upsilon{\partial \psi\over \partial x^\gamma},
$$
\[
R
=
\omega U^\gamma{\partial \theta\over \partial x^\gamma}
+
\tilde\zeta{\partial U^\gamma\over \partial x^\gamma}
+\tilde\iota U^\gamma{\partial \psi\over \partial x^\gamma},
\quad
S_{\alpha\beta}
\equiv
\eta\left({\partial U_\alpha\over \partial x^\beta}
+
{\partial U_\beta\over \partial x^\alpha}
-{2\over 3}g_{\alpha\beta}
{\partial U^\gamma\over \partial x^\gamma}\right),
\]
and
$$ 
-\Delta N^\beta
\equiv
U^\beta \hat P
+\Pi^{\beta\delta}\hat Q_\delta
$$ 
with
$$ 
\hat P
=
\hat\tau U^\delta{\partial \theta\over \partial x^\delta}
+\hat\sigma
{\partial U^\delta\over \partial x^\delta}
+\hat\iota U^\delta{\partial \psi\over \partial x^\delta},
\quad
\hat Q_\delta
\equiv
\hat\nu{\partial \theta\over \partial x^\delta}
+\hat\varsigma U^\epsilon
{\partial U_\delta\over \partial x^\epsilon}
+\hat\upsilon{\partial \psi\over \partial x^\delta}.
$$ 
Our proposed new theory corresponds to choosing 

(i) $\nu=-\tau=-\omega=\chi$,

(ii)
$-\check\varsigma=\sigma$ and $\hat\sigma=-\hat\varsigma=\tilde\sigma$, 

(iii) $\hat\upsilon=-\hat\iota=\mu$, 

(iv) 
$\hat\tau=\hat\nu=\tilde\iota=\check\iota=\upsilon=0$,

i.\ e., 
$$
P
=
-\chi U^\gamma{\partial \theta\over \partial x^\gamma}
+
\sigma{\partial U^\gamma\over \partial x^\gamma}
\quad
Q_\gamma
\equiv
\chi{\partial \theta\over \partial x^\gamma}
-\sigma U^\delta
{\partial U_\gamma\over \partial x^\delta}
$$
\[
R
=
-\chi U^\gamma{\partial \theta\over \partial x^\gamma}
+
\tilde\zeta{\partial U^\gamma\over \partial x^\gamma}
\quad
S_{\alpha\beta}
\equiv
\eta\left({\partial U_\alpha\over \partial x^\beta}
+
{\partial U_\beta\over \partial x^\alpha}
-{2\over 3}g_{\alpha\beta}
{\partial U^\gamma\over \partial x^\gamma}\right),
\]
and
$$ 
\hat P
=
\tilde\sigma
{\partial U^\delta\over \partial x^\delta}
-\mu U^\delta{\partial \psi\over \partial x^\delta},
\quad
\hat Q_\delta
\equiv
-\tilde\sigma U^\epsilon
{\partial U_\delta\over \partial x^\epsilon}
+\mu{\partial \psi\over \partial x^\delta}.
$$
This directly yields 
\eqref{DeltaT}, \eqref{DeltaN}.

As in \cite{FTIG}, we write \eqref{nsf} in 
the {Godunov-Boillat variables} \cite{Go1,Bo,RS}
$$
\psi_\alpha=\frac{U^\alpha}\theta,\quad\psi_4=\psi.
$$ 
Correspondingly, we write the second-order parts of 
$$
-\frac{\partial}{\partial x^ \beta}\left(\Delta T^{\alpha\beta}\right)
\quad\text{and}\quad
-\frac{\partial}{\partial x^ \beta}\left(\Delta N^{\beta}\right)
$$
as 
$$
B^{\alpha\beta c\delta}{\partial^2 \psi_c\over\partial x^\beta \partial x^\delta}
\quad\text{and}\quad
B^{4\beta c\delta}{\partial^2 \psi_c\over\partial x^\beta \partial x^\delta},
$$
respectively, where the index $c$ runs from $0$ through $4$.

Expressing derivatives as
$$ 
{\partial\theta   \over\partial x^\delta}
=\theta^2 U^\gamma
{\partial\psi_\gamma   \over\partial x^\delta},
\quad\quad
{\partial U^\sigma\over\partial x^\delta}
=
\theta \Pi^{\sigma\gamma}
{\partial\psi_\gamma\over\partial x^\delta},
$$ 
we see that 
$$ \begin{aligned}
B^{\alpha\beta\gamma\delta}
=
&+U^\alpha U^\beta(-\chi\theta^2 U^\gamma U^\delta +\sigma\theta \Pi^{\gamma\delta})\\
&+\Pi^{\alpha\beta}(-\chi\theta^2 U^\gamma U^\delta +\tilde\zeta\theta \Pi^{\gamma\delta})\\
&+\chi\theta^2(\Pi^{\alpha\delta}U^\beta+\Pi^{\beta\delta}U^\alpha)U^\gamma\\
&-\sigma\theta(\Pi^{\alpha\gamma}U^\beta+\Pi^{\beta\gamma}U^\alpha)U^\delta\\
&+\eta\theta (\Pi^{\alpha\gamma}\Pi^{\beta\delta}
+\Pi^{\alpha\delta}\Pi^{\beta\gamma}
-(2/3)
\Pi^{\alpha\beta}\Pi^{\gamma\delta})
\end{aligned}
\notag
$$ 
and
$$ 
B^{\alpha\beta 4\delta}=
B^{4\beta\gamma\delta}=
0.
$$ 
as well as
$$ 
B^{4\beta 4\delta}= 
-\mu U^\beta U^\delta
+\mu \Pi^{\beta\delta}.
$$ 

With any $N^\beta$ satisfying
$
N^\beta U_\beta=0,\quad N^\beta N_\beta=1,
$
the rest-frame coefficient matrices 
$
B^{a\beta c\delta}U_\beta U_\delta
$
and
$
B^{a\beta c\delta}N_\beta N_\delta
$ 
are thus given by
\be\label{HKMmatrices}
\begin{pmatrix}
-\chi\theta^2&0&0\\
0&-\sigma\theta\delta^{ij}&0\\
0&0&-\mu
\end{pmatrix},
\quad\quad
\begin{pmatrix}
\chi\theta^2&0&0\\
0&
\eta\theta \delta^{ij}
+(\frac13\eta+\tilde\zeta)\theta N^i N^j&0\\ 
0&0&\mu
\end{pmatrix}.
\ee
This confirms the (covariant version of the) HKM definiteness conditions
(\cite{FTPR}, (4.1)-(4.5))
$$ 
B^{a\beta c\delta}H_\beta H_\delta V_a V_c<0
\ \ \hbox{for all $V_a\neq 0$,}
\label{pos1}
$$  %
and 
$$ 
B^{a\beta c\delta}N_\beta N_\delta V_a V_c>0
\ \ \hbox{for all $V_a\neq 0$,}
\label{pos2}
$$ 
for 
$$ 
\hbox{some $H_\beta$ with $H_\beta H^\beta<0$ and 
all $N_\beta\neq0$ with
$N_\beta H^\beta=0$.}
\label{sh} 
$$ 
The statements on causality and sharp causality follow as in \cite{FTIG} from \eqref{HKMmatrices}
and the facts that causality and sharp causality correspond to $\tilde\zeta\ge-\frac13\eta, \sigma\ge \eta$ and 
$\tilde\zeta=-\frac13\eta,\sigma=\eta$, respectively.  

\section{Entropy production}
\setcounter{equation}0
In the Eckart frame, the entropy production is classically known (cf.\ \cite{W}, p.\ 55) as  
\begin{align*}
\mathcal Q
\equiv
&  -{1\over\theta^2}{\partial\theta\over\partial x^0}\Delta T^{00}|_0
-{1\over\theta^2}\left({\partial\theta\over\partial x^i}+\theta{\partial 
u_i\over \partial x^0}\right)\Delta T^{i0}|_0
-{1\over\theta}{\partial u_i\over\partial x^j}\Delta T^{ij}|_0
-{\partial\psi\over\partial x^0}\Delta N^{0}|_0
-{\partial \psi\over\partial x^j}\Delta N^{j}|_0
\notag
\\
=&
{\chi\over\theta^2}
|\nabla\theta+\theta\dot\bu|^2
+ \frac{\eta}{2\theta}||\mathcal S\bu||^2+\frac \zeta\theta(\nabla\cdot\bu)^2+\mu|\nabla\psi|^2\ge 0.
\end{align*}
The claim of Theorem 1 on entropy production is an immediate consequence of the following result, 
which is not restricted to polytropic gases.
\begin{prop}
Let the dissipation coefficients be of magnitude $O(\epsilon)$.  
Under first-order equivalence transformations, the entropy production then changes by a difference 
$\Delta\mathcal Q$ of higher order $O(\epsilon^2)$.   
\end{prop}
\begin{proof}
This is obvious for gradient reexpressions. For any velocity shift  
$$
\begin{pmatrix}
*&*\\
*&*\\
*&*
\end{pmatrix}
\to
\begin{pmatrix}
*&*+\Delta\bu^\top\\
*+\Delta\bu,&*\\
*,&*+(1/h)\Delta \bu^\top
\end{pmatrix}
$$
with $\Delta\bu=O(\epsilon)$, 
we find 
$$
\Delta \mathcal Q
=\frac1{\theta^2}(\nabla\theta+\theta\dot\bu)\cdot\Delta \bu+\frac1h\nabla\psi\cdot\Delta\bu=O(\epsilon^2),
$$
where we have used (cf.\ \cite{FTIG}, eq.\ (4.10)) that  
$$
\frac1{\theta^2}(\nabla\theta+\theta\dot\bu)+\frac1h\nabla\psi=O(\epsilon). 
$$
Finally, consider any thermodynamic shift 
$$
\begin{pmatrix}
*&*\\
*&*\\
*&*
\end{pmatrix}
\to
\begin{pmatrix}
*+\Delta\rho&*\\
*&*+\Delta p\,\bf I\\
*+\Delta n&*
\end{pmatrix}
$$
with a triple $(\Delta \rho, \Delta n,\Delta p)=O(\epsilon)$ that is compatible with the 
equation of state. We assume the latter to be given as   
$$
p=p(\theta,\psi),\quad\text{so that \ }\rho=\theta p_\theta-p,\ \ n=p_\psi/\theta,
$$
introduce $\Delta\rho,\Delta n$ and rewrite $\dot\theta,\dot\psi$ as
$$
\begin{pmatrix}
\Delta\rho \\ \Delta n
\end{pmatrix}
=
A
\begin{pmatrix}
\Delta\theta \\ \Delta \psi
\end{pmatrix},
\quad
\begin{pmatrix}
\dot\rho \\ \dot n
\end{pmatrix}
=
A
\begin{pmatrix}
\dot\theta \\ \dot\psi
\end{pmatrix},
$$ 
with 
$$
A=
\begin{pmatrix}
\rho_\theta&\rho_\psi\\
n_\theta&n_\psi 
\end{pmatrix}
=
\begin{pmatrix}
\theta p_{\theta\theta}&\theta p_{\theta\psi}-p_\psi\\
\theta^{-2}(\theta p_{\theta\psi}-p_\psi)&\theta^{-1}p_{\psi\psi} 
\end{pmatrix},
$$ 
and evaluate the change in entropy production as
\begin{align*}
\Delta\mathcal Q&=\frac1{\theta^2}\dot\theta\Delta\rho+\frac1\theta \nabla\cdot\bu\Delta p +\dot\psi\Delta n\\
&=\nabla\cdot\bu
\bigg(- \frac{\rho+p}{\theta^2}\Delta\theta+\frac1\theta\Delta p-n\Delta\psi\bigg)+O(\epsilon^2)\\
&=O(\epsilon^2),
\end{align*}
where we have used equations (6.4) of \cite{FTIG} and the compatibility relation 
$$
\Delta p=p_\theta\Delta\theta+p_\psi\Delta \psi=\frac{\rho+p}\theta\Delta\theta+n\theta \Delta\psi. 
$$
\end{proof}


\end{document}